\documentclass[times]{elsarticle}

\usepackage[utf8]{inputenc}
\usepackage[T1]{fontenc}
\usepackage{amsmath,amsthm}
\usepackage{xspace}

\DeclareFontFamily{U}{mathb}{\hyphenchar\font45}
\DeclareFontShape{U}{mathb}{m}{n}{ <-6> mathb5 <6-7> mathb6 <7-8>
  mathb7 <8-9> mathb8 <9-10> mathb9 <10-12> mathb10 <12-> mathb12 }{}
\DeclareSymbolFont{mathb}{U}{mathb}{m}{n}
\DeclareMathSymbol{\prec}{\mathrel}{mathb}{"A0}
\DeclareMathSymbol{\succ}{\mathrel}{mathb}{"A1}
\DeclareMathSymbol{\preceq}{\mathrel}{mathb}{"A8}
\DeclareMathSymbol{\succeq}{\mathrel}{mathb}{"A9}
\DeclareMathSymbol{\precneq}{\mathrel}{mathb}{"AC}
\DeclareMathSymbol{\succneq}{\mathrel}{mathb}{"AD}

\usepackage[final]{hyperref}
\usepackage[final]{microtype}

\usepackage{ifdraft}

\usepackage{color}
\definecolor{gray}{gray}{0.4}

\usepackage{amscd}
\usepackage{natbib}
\usepackage{enumerate}

\usepackage[font=small]{caption}
\usepackage{subcaption}

\usepackage[english]{babel}
\usepackage[shortlabels,inline]{enumitem}

\usepackage[algoruled,vlined,english,linesnumbered]{algorithm2e}
\SetArgSty{textup}
\SetKwInOut{Input}{Input}\SetKwInOut{Output}{Output}
\SetKwIF{If}{ElseIf}{Else}{if}{:}{else if}{:}{end}%
\SetKw{KwTo}{in}\SetKwFor{For}{for}{\string:}{end}%
\SetKwFor{While}{while}{:}{end}%
\SetKwFor{Forall}{forall}{:}{end}%
\AlgoDontDisplayBlockMarkers
\SetAlgoNoEnd
\makeatletter
\newcommand{\removelatexerror}{\let\@latex@error\@gobble}
\makeatother

\usepackage{tikz}

\usepackage{fixme}
\fxsetup{author={T},envlayout={color},%
  layout={footnote,marginclue},theme=color}

\FXRegisterAuthor{ch}{ach}{C}

\usepackage[mathscr]{euscript}

\usepackage{array}
\usepackage{mathtools}
\mathtoolsset{showonlyrefs,showmanualtags}

\usepackage{tabu}
\usepackage{multirow}
\usepackage{booktabs}

\usepackage{siunitx}

\theoremstyle{theorem}
\newtheorem{theorem}{Theorem}[section]
\newtheorem{lemma}[theorem]{Lemma}
\newtheorem{proposition}[theorem]{Proposition}

\newtheorem{corollary}[theorem]{Corollary}

\theoremstyle{definition}
\newtheorem{definition}[theorem]{Definition}
\newtheorem{problem}[theorem]{Problem}

\theoremstyle{remark}
\newtheorem{example}[theorem]{Example}  
\newtheorem{remark}[theorem]{Remark}

\newcommand{\lm}{\mathrm{lm}}
\newcommand{\degdiff}{\mathrm{degdiff}}
\newcommand{\degmin}{\mathrm{degmin}}

\newcommand{\QQ}{\mathbb{Q}}
\newcommand{\NN}{\mathbb{N}}
\def\Fr{\Sigma}
\def\MFr{M(\Fr)}
\newcommand{\KX}{K\langle X\rangle}

\def\Sig{^{[\Sigma]}}

\newcommand{\sig}{\mathfrak{s}}

\renewcommand{\epsilon}{\varepsilon}

\newcommand{\myspan}{\mathrm{span}}

\newcommand{\Syz}{\textnormal{Syz}}
\newcommand{\supp}{\textnormal{supp}}
\newcommand{\rank}{\textnormal{rank}}

\def\<#1>{\langle#1\rangle}

\newcommand{\eg}{e.g.}
\newcommand{\ie}{i.e.\xspace}

\newcommand{\NP}{\textup{\textsf{NP}}}

\begin{document}

\begin{frontmatter}

\title{Short proofs of ideal membership}

 \author[1]{Clemens Hofstadler}
\ead{clemens.hofstadler@mathematik.uni-kassel.de}

\author[2]{Thibaut Verron}
\ead{thibaut.verron@gmail.com}

\affiliation[1]{
	organization={Institute of Mathematics, University of Kassel}, 
	addressline={Heinrich-Plett-Stra\ss e 40},
	city={Kassel}, 
	country={Germany}
	}
\affiliation[2]{
	organization={Institute for Algebra, Johannes Kepler University},
	addressline={Altenberger Stra\ss e 69},
         city={Linz},
         country={Austria}
      	}


\begin{abstract}
A cofactor representation of an ideal element, that is, a representation in terms of the generators, can be considered as a certificate for ideal membership.
Such a representation is typically not unique, and some can be a lot more complicated than others.
In this work, we consider the problem of computing sparsest cofactor representations, \ie, representations with a minimal number of terms, of a given element in a polynomial ideal.
While we focus on the more general case of noncommutative polynomials, all results also apply to the commutative setting.

We show that the problem of computing cofactor representations with a bounded number of terms is decidable and $\NP$-complete.
Moreover, we provide a practical algorithm for computing sparse (not necessarily optimal) representations by translating the problem into a linear optimization problem
and by exploiting properties of signature-based Gr\"obner basis algorithms.
We show that, for a certain class of ideals, representations computed by this method are actually optimal, and we present experimental data illustrating that it can lead to noticeably sparser cofactor representations.
\end{abstract}

\end{frontmatter}

\section{Introduction}
\label{sec:introduction}

%

In polynomial algebra, the ideal membership problem is one of the most fundamental problems with many important applications from polynomial system solving over polynomial identity testing to automated reasoning.
In case of noncommutative polynomials in the free algebra, this last application is particularly relevant and has been the focus of recent research~\cite{HW94,SL20,CHRR20,raab2021formal,bernauer2023automatise,hofstadler-phd}.
Several frameworks have been developed that allow to prove the correctness of statements about linear operators (such as matrices, homomorphisms, bounded operators, etc.) by verifying ideal membership of noncommutative polynomials.
In this setting, any cofactor representation, that is, any representation of an ideal element as a linear combination of the generators, can be considered as a proof of the corresponding operator statement, and noncommutative Gr\"obner bases can be used to compute such representations~\cite{HRR19}, see also~\cite{letterplace} and references therein for available software.

In general, cofactor representations are not unique and different representations can differ drastically in their complexity. 
We could observe empirically that representations computed by Gr\"obner bases are often significantly longer than necessary.
In this work, we discuss the problem of finding \emph{sparsest} cofactor representations of an ideal element, that is,
representations with a minimal number of terms.
We focus on the situation of noncommutative polynomials, as we are particularly interested in computing short proofs of operator statements.
However, all techniques also apply analogously to commutative polynomials.

Although ideal membership in the free algebra is only semidecidable, we show that the problem of computing cofactor representations with the number of terms bounded by $N \in \NN$ is decidable, yet \NP-complete.
This yields a first (impractical) algorithm for computing sparsest cofactor representations (Algorithm~\ref{algo:decidability}).

We then describe how to obtain a practical algorithm for computing sparse (not necessarily sparsest) representations by making two simplifications:
\begin{enumerate}
	\item We restrict the search space to a finite dimensional subspace by only considering cofactor representations with terms smaller than a designated bound.\label{item:simplification-bound}
	\item We use the sum of the absolute values of the coefficients, i.e., the $\ell_1$-norm, of a representation as a complexity measure.\label{item:simplification-coeffs}
\end{enumerate}
With these simplifications, we translate the problem of finding sparse cofactor representations into solving a linear programming problem.
Our main result is Algorithm~\ref{algo:short-rep}, which computes, starting from any given cofactor representation, a minimal one w.r.t.\ the conditions~\eqref{item:simplification-bound} and~\eqref{item:simplification-coeffs}.

We also show that the second simplification does in fact impose no restriction for a class of ideals that appears frequently when translating operator statements.
In particular, we prove that, under certain assumptions satisfied by most examples studied in practice, Algorithm~\ref{algo:short-rep} computes a sparsest representation among all representations satisfying condition~\eqref{item:simplification-bound}.
Finally, we demonstrate the effectiveness of Algorithm~\ref{algo:short-rep} on several examples coming from actual operator statements.

Our algorithm relies on information provided by noncommutative signature-based Gr\"obner basis algorithms.
Initially developed in the commutative setting to improve Gr\"obner basis computations, signature-based algorithms have been subject of extensive research, see~\cite{EF17} and~\cite{Lairez-2022-AxiomsForTheory}.
Recently, the authors proposed a generalization of these algorithms to free algebras over fields~\cite{Hofstadler-2022-SignatureGroebnerBases} and rings~\cite{mixed-algebra}.

Signature-based algorithms compute, in addition to a Gröbner basis, some information on how the polynomials in that basis were computed.
This additional information allows the algorithms not only to predict and avoid redundant computations, but also to compute a Gr\"obner basis of the syzygy module of the generators.
In particular, signature-based algorithms compute this basis in a very structured way, and precisely this structure is what we exploit in our algorithm.


To the best of our knowledge, the problem at hand has not been studied in this form, that is, for noncommutative polynomials and from a practical point of view.
For commutative polynomials, the complexity of the ideal membership problem has been the focus of extensive research.
Starting from early work showing that the problem requires, in general, doubly exponential time~\cite{mayr-meyer-82,mayr89}, to recent results identifying special classes of polynomial ideals, often encompassing systems arising in applications in cryptography, physics or combinatorics, for which it can be solved more efficiently~\cite{Faugere-2012-group,Faugere-2013-gen-minrank,Faugere-2014-sparse,Faugere-2016-whomo-2,mayr2017,mastrolilli2021complexity,bulatov2022complexity,bulatov2022ideal}. 
The (commutative) ideal membership problem is also central to the field of algebraic proof complexity, where finding sparse cofactor representations corresponds to finding \emph{short Nullstellensatz proofs}, see~\cite{pitassi2016algebraic} for a recent survey.
In this setting, plenty of work is dedicated to establishing degree bounds of commutative cofactor representations in the form of effective versions of Hilbert's Nullstellensatz~\cite{Kol88,Jel05,d2013heights}.
It is worth noting that the research listed above primarily concentrates on theoretical aspects without considering the practical applicability of the developed methods and algorithms.
In contrast, in this work, we focus on actual computability of sparse representations for concrete examples.  
Proof simplification in automated theorem proving in general is also addressed in \eg,~\cite{veroff2001finding,kinyon2019proof}.


%

\section{Preliminaries}
\label{sec:non-comm-grobn}

We recall basic definitions regarding noncommutative polynomials and signature-based Gröbner basis algorithms in the free algebra.
For more details, we refer to \cite{Hofstadler-2022-SignatureGroebnerBases,mixed-algebra}.
For an introduction to noncommutative Gr\"obner bases without signatures, see~\cite{Xiu12,Mor16}.

\subsection{Free algebra}
\label{sec:free-algebra}
Throughout this paper, $X = \{x_{1},\dots,x_{k}\}$ is a finite set of indeterminates and $\langle X \rangle$ is the free monoid over $X$.
For $m \in \<X>$, we denote by $|m|$ the length of $m$.
Let $K$ be a field and $K\< X>$ be the free algebra over $X$.
We consider the elements in $K\<X>$ as \emph{noncommutative polynomials}, 
where monomials are given by words over $X$ and multiplication is given by concatenation of words. 
Note that $K\<X>$ is not Noetherian if $|X| > 1$.
For $F \subseteq K\<X>$, we denote by $(F)$ the (two-sided) ideal generated by $F$.

A \emph{monomial ordering} on $\<X>$ is a well-ordering $\leq$ compatible with the multiplication in $\<X>$, that is,
$m \leq m'$ implies $amb \leq am'b$ for all $a,b,m,m'\in\<X>$.
We fix a monomial ordering $\leq$.

\begin{example}
The degree lexicographic ordering $\leq_{\textup{deglex}}$ is a monomial ordering on $\<X>$ 
where two words $m,m' \in \<X>$ are first compared by their lengths and ties are broken by comparing the variables in $m$ and $m'$ from left to right using the lexicographic ordering $x_1 <_{\textup{lex}} \dots <_{\textup{lex}} x_k$.
\end{example}

For $f \in K\<X>$, the support of $f$, denoted by $\supp(f)$, is the set of all monomials appearing in $f$.
For $f \neq 0$, the leading monomial $\lm(f)$ of $f$ is the maximal element w.r.t.\ $\leq$ in $\supp(f)$ and
 the degree $\deg(f)$ of $f$ is $\deg(f)= \max_{m \in \supp(f)} |m|$.
Additionally, we set $\deg(0) = -1$.
 
\subsection{Free bimodule}
\label{sec:free-bimodule}
We fix a family of polynomials $f_{1},\dots,f_{r} \in \KX$, generating an ideal $I = (f_{1},\dots,f_{r})$.
We extend the previous definitions to the \emph{free $K\<X>$\nobreakdash-bimodule} $\Fr$ of rank $r \in \NN$, given by $\Fr = (K\<X> \otimes K\<X>)^r$.
The canonical basis of $\Fr$ is $\varepsilon_1,\dots,\varepsilon_r$,
where $\varepsilon_i = (0,\dots,0,1 \otimes 1,0,\dots,0)$ with $1 \otimes 1$ appearing in the $i$th position for $i=1,\dots,r$.

The set $\MFr$ of \emph{module monomials} in $\Fr$ is given by $\MFr = \{a\varepsilon_ib \mid a,b \in \<X>,  i =1,\dots, r\}$.
Every element $\alpha \in \Sigma$ has a unique representation of the form $\alpha = \sum_{i=1}^d c_i a_i \varepsilon_{j_i} b_i$ with nonzero $c_i \in K$ and pairwise different $a_i \varepsilon_{j_i} b_i \in \MFr$.
We denote its support by $\supp(\alpha) = \{a_i \varepsilon_{j_i} b_i\mid i = 1,\dots,d\}$ and associate to it the polynomial $\overline \alpha \coloneqq \sum_{i=1}^d c_i a_i f_{j_i} b_i \in I$.
With this, the (weighted) degree of $\alpha$ is $\deg(\alpha) = \max_i \deg(a_i f_{j_i} b_i)$.

A \emph{module ordering} on $M(\Sigma)$ is a well-ordering $\preceq$ compatible with scalar multiplication, that is, $\mu \preceq \mu'$ implies $a\mu b \preceq a \mu' b$ for all $\mu,\mu' \in M(\Sigma)$ and $a,b \in \<X>$.
We fix a module ordering $\preceq$.

\begin{example}
The degree-over-position-over-term ordering $\preceq_{\textup{DoPoT}}$ is a module ordering
where two module monomials $a \varepsilon_i b, a' \varepsilon_j b' \in \MFr$ are first compared by their degrees
and ties are broken by comparing the tuples $(i,a,b)$ and $(j,a',b')$ lexicographically using a monomial ordering for the monomial comparisons.
\end{example}

\begin{definition}
The \emph{signature} $\sig(\alpha)$ of a nonzero $\alpha \in \Fr$ is the maximal element w.r.t.\ $\preceq$ in $\supp(\alpha)$.
\end{definition}


%
%

\subsection{Signature-based algorithms}
\label{sec:sig-algos}

For the rest of this work, we assume that the monomial ordering $\leq$ and the module ordering $\preceq$ satisfy:
\begin{itemize}
	\item $\leq$ and $\preceq$ are \emph{compatible} in the sense that $a < b$ iff $a \varepsilon_i \prec b \varepsilon_i$ iff $\varepsilon_i a \prec \varepsilon_i b$
	for all $a,b \in \<X>$ and $i = 1,\dots,r$;
	\item $\preceq$ is \emph{fair}, meaning that the set $\{ \mu' \in \MFr \mid \mu' \prec \mu\}$ is finite for all $\mu \in \MFr$;
\end{itemize}
	
\begin{example}
The module ordering $\preceq_{\textup{DoPoT}}$ is fair and compatible with $\leq_{\textup{deglex}}$.
\end{example}

A \emph{labeled polynomial} $f^{[\alpha]}$ is a pair $(f,\alpha) \in I \times \Sigma$ with $f = \overline{\alpha}$. 
We call the set $I^{[\Sigma]} \coloneqq \{ f^{[\alpha]} \mid f \in I, \overline \alpha = f\}$ the \emph{labeled module} generated by $f_1,\dots,f_r$.
It forms a $K\<X>$\nobreakdash-subbimodule of $I \times \Sigma$ with component-wise addition and scalar multiplication.

A syzygy of $I^{[\Sigma]}$ is an element $\gamma \in \Sigma$ such that $\overline{\gamma}=0$.
The set of syzygies of $I^{[\Sigma]}$, denoted by $\Syz(I^{[\Sigma]})$, forms a $K\<X>$-subbimodule of $\Sigma$.
We recall the notion of Gr\"obner basis of $\Syz(I^{[\Sigma]})$.

\begin{definition}
A set $H \subseteq \Syz(I^{[\Sigma]})$ is a \emph{Gr\"obner basis of $\Syz(I^{[\Sigma]})$ (up to signature $\sigma \in \MFr$)} if 
for all nonzero $\gamma \in \Syz(I^{[\Sigma]})$ (with $\sig(\gamma) \prec \sigma$), 
there exist $d \in \NN$ and $\gamma_i \in H$, $c_i \in K$, $a_i,b_i \in \<X>$
such that $\gamma = \sum_{i=1}^d c_i a_i \gamma_i b_i$ and $\sig(a_i \gamma_i b_i) \preceq \sig(\gamma)$ for all $i = 1,\dots,d$.
\end{definition}

In general, a syzygy module need not have a finite Gr\"obner basis.
Nevertheless, Gröbner bases of $\Syz(I\Sig)$ can be enumerated by increasing signatures using signature-based algorithms.

\begin{theorem}\label{thm:syz-basis}
  There exists an algorithm to correctly enumerate a Gröbner basis of $\Syz(I\Sig)$ in increasing signature order.
  In particular, for all $\sigma \in \MFr$, stopping the algorithm at the first syzygy with signature $\succeq \sigma$ yields a finite Gröbner basis of $\Syz(I\Sig)$ up to signature $\sigma$.
\end{theorem}
\begin{proof}
  The algorithm is~\cite[Algo.~1]{Hofstadler-2022-SignatureGroebnerBases}.
  Its correctness is proved in~\cite[Thm.~42]{Hofstadler-2022-SignatureGroebnerBases}
  and the fact that the signatures increase throughout is proved in~\cite[Lem.~43]{Hofstadler-2022-SignatureGroebnerBases}.
\end{proof}

The algorithm uses signature-based Gröbner techniques, and also enumerates a (possibly infinite) \emph{labeled Gröbner basis} of $I\Sig$.
The definition and construction of labeled Gröbner bases are beyond the scope of this paper, but we note that, among their properties,
a labeled Gröbner basis is a set $G\Sig \subseteq I\Sig$ such that the set $\{f \mid f^{[\alpha]} \in G\Sig\}$ is a Gröbner basis of $I$.
In particular, given $f \in I$, performing polynomial reductions by $G\Sig$ to reduce $f$ to $0$ and adding the labeling of the reducers yields a cofactor representation of $f$ w.r.t. $f_1,\dots,f_{r}$.

\begin{remark}
The algorithm mentioned in the proof of Theorem~\ref{thm:syz-basis} is inefficient as it has to perform a lot of expensive module arithmetic.
A more efficient way of realizing Theorem~\ref{thm:syz-basis} is to combine~\cite[Algo.~2,4]{Hofstadler-2022-SignatureGroebnerBases} for computing \emph{signature Gröbner bases}.
These algorithms work with pairs $(f,\sig(\alpha))$ instead of labeled polynomials $f^{[\alpha]}$ and reconstruct the full module representations \emph{a posteriori}, avoiding a lot of module arithmetic in this way.
\end{remark}

A finite Gr\"obner basis of $\Syz(I^{[\Sigma]})$ up to signature $\sigma$, as output by the signature-based algorithm, provides an effective description of all syzygies with signature $\prec \sigma$.
This is the crucial property of signature-based algorithms that we exploit in Section~\ref{sec:finding-short-proofs} to compute sparse cofactor representations.

\section{Decidability and complexity}
\label{sec:decid-compl-short}

For the rest of this section, we fix a family of polynomials $f_{1},\dots,f_{r} \in K\<X>$ generating an ideal $I$.
Any cofactor representation of an ideal member $f \in I$ can be identified with an element $\alpha \in \Fr$ such that $\overline \alpha = f$.
The \emph{weight} of such a representation is given by the $\ell_0$-``norm'' $\lVert\alpha\rVert_0 \coloneqq |\supp(\alpha)|$. 
Then, with the set $R(f) \coloneqq \{ \alpha \in \Fr \mid \overline \alpha = f\}$ of \emph{(cofactor) representations} of $f$,
a \emph{sparsest} (cofactor) representation of $f$ corresponds to a minimal element w.r.t.\ $\lVert\cdot\rVert_0$ in $R(f)$.
We denote the set of all such minimal elements by $R_0(f)$.
If $f \notin I$, we set $R(f) = R_0(f) = \emptyset$.

\begin{remark}
Ideal membership in the free algebra is only semidecidable.
Consequently, we can also not decide whether $R(f) = \emptyset$ or not.
\end{remark}

\begin{remark}
The function $\lVert\cdot\rVert_0$ is not a norm as it is not homogeneous,
but one can associate to it a metric called \emph{Hamming distance}.
\end{remark}

In the following, we study the decidability and complexity of computing cofactor representations of bounded weight.
To this end, we assume that the coefficient field $K$ is computable, in the sense that the basic arithmetic operations as well as equality testing are effective. 
This means in particular that linear systems can be solved effectively using, for example, Gaussian elimination.  
With this, we consider the following problem.

\begin{problem}[Sparse cofactor representation]${}$\label{def:shortest_repr_problem}\\
  \textsc{Input}: $f, f_1,\dots,f_{r} \in K\langle X\rangle$, $N \in \NN$\\
  \textsc{Output}: a cofactor representation $\alpha \in R(f)$ with $\lVert \alpha \rVert_{0} \leq N$ if one exists, otherwise \textsf{False}.
\end{problem}

We show that Problem~\ref{def:shortest_repr_problem} is decidable, and we give an algorithm reducing it to Problem~\ref{def:sparse_linsolve} below of finding sparse solutions of linear systems, formally also known as the Min-RVLS (MINimum Relevant Variables in Linear System) problem.
This not only yields an algorithm for computing sparsest cofactor representations if $f \in I$, but it, in principle, also provides a semidecision procedure for ideal membership.
We focus on the first application here.

\begin{problem}[Sparse solution of linear system {[Min-RVLS]}]${}$\label{def:sparse_linsolve}\\
  \textsc{Input}: $A \in K^{m\times n}, \mathbf{b}\in K^{m}$, $N \in \{0,\dots,n\}$\\
  \textsc{Output}: a vector $\mathbf{y} \in K^{n}$ with $A\mathbf{y} = \mathbf{b}$ and $\lVert \mathbf{y} \rVert_{0} \leq N$ if one exists, otherwise \textsf{False}.
\end{problem}

Problem~\ref{def:sparse_linsolve} arises in many areas~\cite{chen2001atomic,candes2005decoding,donoho2006compressed}.
It is clearly decidable, for instance by looping over all $N$-subsets of $\{1,\dots,n\}$ for possible sets of nonzero coefficients of solutions.
Furthermore, in many cases, most notably for $K = \QQ$, it is known to be \NP-hard, with the corresponding decision problem being \NP\nobreakdash-complete~\cite[Problem MP5]{Garey-1990-ComputersIntractabilityGuide}.

In order to reduce the sparse cofactor representation problem to linear algebra, we need to constrain the solutions to a finite dimensional vector space, which requires to bound the degree of a solution.
The degree of elements in $R(f)$ can be arbitrarily large, but we can bound the degree of \emph{minimal} representations.
A cofactor representation $\alpha = \sum_{i=1}^{d} c_{i}a_{i}\varepsilon_{j_{i}}b_{i} \in R(f)$ is \emph{minimal} if no sub-sum is a syzygy, that is, $\sum_{i \in J} c_i a_i f_{j_i} b_i \neq 0$ for all non-empty subsets $J \subseteq \{1,\dots,d\}$.
To obtain the degree bound, we recall the notion of \emph{(polynomial) rewriting} introduced in~\cite[Def.~2]{raab2021formal}.

\begin{definition}
Let $f,g \in K\<X>$ and $a,b \in \<X>$ such that $\supp(f) \cap \supp(agb) \neq \emptyset$.
For every $c \in K$, we say that $f$ \emph{can be rewritten to $f + cagb \in K\<X>$ by $g$}.

Furthermore, we say that $f$ \emph{can be rewritten to $h$ by} $G \subseteq K\<X>$ if there are $h_0,\dots,h_d \in K\<X>$,
$h_d = f$, $h_0 = h$ and $g_1,\dots,g_d \in G$ such that $h_{k}$ can be rewritten to $h_{k-1}$ by $g_k$ for all $k = 1,\dots,d$.
\end{definition}

Rewriting can be considered as a weaker form of polynomial reduction, not requiring that a polynomial gets ``simplified'' by a rewriting step. 
Nevertheless, $f$ can be rewritten to zero by $\{f_1,\dots,f_r\}$ if and only if $f \in I$, see~\cite[Lem.~4]{raab2021formal}.
More importantly, we can show that any minimal representation of $f$ can be obtained by rewriting $f$ to $0$ by $\{f_1,\dots,f_r\}$ and logging the rewriting steps.

\begin{lemma}\label{lemma:rewriting-short-rep}
Let $f \in (f_1,\dots,f_r)$ and $\alpha = \sum_{i=1}^{d} c_{i}a_{i}\varepsilon_{j_{i}}b_{i} \in R(f)$ be a minimal representation of $f$.
Furthermore, for $k =0,\dots,d$, let $h_{k} = \sum_{i=1}^{k} c_{i}a_{i}f_{j_{i}}b_{i}$.
In particular, $h_{d}=f$ and $h_{0} = 0$.
Then, possibly after reordering the summands of $\alpha$, $h_k$ can be rewritten to $h_{k-1}$ by $\{f_1,\dots,f_r\}$ for all $k = 1,\dots,d$.
 \end{lemma}
 
\begin{proof}
  We perform induction on the weight $d$ of a minimal representation of $f$.
  For $d=0$ there is nothing to prove.
  Assume now that $d>0$ and that the result is proven for polynomials with a minimal representation of weight $d-1$.
  Because $\alpha$ is a minimal representation, $f$ cannot be $0$.
  Since $f = \overline \alpha = \sum_{i=1}^{d} c_{i}a_{i}f_{j_{i}}b_{i}$, the support of $f$ is contained in the union of the supports of the $a_{i}f_{j_{i}}b_{i}$, and there exists $1 \leq k \leq d$ such that $\supp(f) \cap \supp(a_{k}f_{j_{k}}b_{k}) \neq \emptyset$.
  Possibly after reordering the summands of $\alpha$, we can assume $k = d$.
  So $f$ can be rewritten to $h_{d-1} = f - c_{d}a_{d}f_{j_{d}}b_{d}$ using $f_{j_d} \in \{f_1,\dots,f_r\}$.
  Furthermore, $h_{d-1} = \sum_{i=1}^{d-1} c_{i}a_{i}f_{j_{i}}b_{i}$ has a minimal representation of weight $d-1$, because adding one term results in a minimal representation of weight $d$ for $f$.
  So, by induction hypothesis, this representation of $h_{d-1}$ is (up to reordering of the summands) a sequence of rewritings by $\{f_1,\dots,f_r\}$.
  \end{proof}

Another crucial property of rewriting is that we can bound the degree of the output in terms of the degree of the input and 
the \emph{degree difference} of the rewriter.
The \emph{degree difference} $\degdiff(g)$ of a nonzero $g \in K\<X>$ is $\degdiff(g) = \deg(g) - \degmin(g)$,  where $\degmin(g)= \min_{m \in \supp(g)} |m|$.

\begin{lemma}\label{lemma:rewriting-bound}
Let $f,g \in K\<X>$ and $c \in K$, $a,b \in \<X>$.
If $f$ can be rewritten to $h = f + cagb$ by $g$, then $\max\{ \deg(h),\deg(agb)\}  \leq \deg(f) + \degdiff(g)$. 
\end{lemma}
\begin{proof}
  By definition, there exists a monomial $m$ in $\supp(g)$ such that $amb \in \supp(f)$, so $|amb| \leq \deg(f)$, or equivalently $|a|+|b| \leq \deg(f)-|m|$.
  Since $m \in \supp(g)$, $|m| \geq \degmin(g)$, and all in all, 
  \begin{align}
    \label{eq:2}
    &\deg(agb) = |a| + |b| + \deg(g) \\
    \leq &\deg(f) - \degmin(g) + \deg(g) = \deg(f) + \degdiff(g).
  \end{align}
  As $\deg(h) \leq \max\{ \deg(f),\deg(agb)\}$ and $\degdiff(g) \geq 0$, we conclude that also $\deg(h) \leq \deg(f) + \degdiff(g)$.
\end{proof}

Combining Lemma~\ref{lemma:rewriting-short-rep} and~\ref{lemma:rewriting-bound}, we obtain a bound on the degree of minimal cofactor representations of $f$ of bounded weight.
Since any sparsest cofactor representation is, in particular, minimal, this also yields a bound on the degree of sparsest representations.  

\begin{corollary}
  \label{prop:bound-shortest-repr}
Let $f \in (f_1,\dots,f_r)$, $N \in \NN$ and $\alpha \in R(f)$ be a minimal representation of $f$.
If $\lVert \alpha \rVert_0 \leq N$, then
$\deg(\alpha) \leq  \deg(f) + N \max_i \degdiff(f_i)$.
\end{corollary}

\begin{proof}
Write $\alpha$ as $\alpha = \sum_{i=1}^d c_i a_i \varepsilon_{j_i} b_i$.
By definition, $\deg(\alpha) = \max_i \deg(a_i f_{j_i} b_i)$, and, according to Lemma~\ref{lemma:rewriting-short-rep}, each $a_i f_{j_i} b_i$ is a rewriter in a rewriting sequence from $f$ to $0$.
Thus, Lemma~\ref{lemma:rewriting-bound} shows inductively that $\deg(\alpha) \leq \deg(f) + \lVert \alpha \rVert_0 \max_i \degdiff(f_i)$ and the result follows since $\lVert \alpha \rVert_0 \leq N$. 
\end{proof}

As a consequence, we can state an algorithm for computing a cofactor representation of weight bounded by $N \in \NN$, reducing to the problem of finding a sparse solution of a linear system.

\begin{algorithm}
  \SetAlgoLined
  \KwIn{$f, f_1,\dots,f_{r} \in K\langle X\rangle$, $N \in \NN$}
  \KwOut{$\alpha \in R(f)$ with $\lVert \alpha \rVert_0 \leq N$ if one exists, otherwise \textsf{False}}
  $D \leftarrow \deg(f) + N \max_i \degdiff(f_i)$ \;
  $L \leftarrow \{a f_i b \mid a,b\in \<X>, \deg(af_i b) \leq D, i = 1,\dots,r\}$ \;
  \Return a $K$-linear combination of elements of $L$ equal to $f$ with $\leq N$ nonzero summands if one exists, otherwise \textsf{False}\;
  \caption{Sparse cofactor representation}
  \label{algo:decidability}
\end{algorithm}


\begin{corollary}
  Algorithm~\ref{algo:decidability} terminates and is correct.
\end{corollary}
\begin{proof}
  The algorithm reduces the problem to that of finding sparse solutions of a linear system.
  This problem is decidable (recall that $K$ is computable), so the algorithm terminates.
  
  There exists a representation of $f$ of weight $\leq N$ if and only if there exists a \emph{minimal} representation $\alpha$ of $f$ of weight $\leq N$.
  By Corollary~\ref{prop:bound-shortest-repr}, this representation is given by a linear combination of weight $\lVert \alpha \rVert_{0} \leq N$ consisting of elements of $L$.
  So the algorithm is correct.
\end{proof}

It is also possible to describe a reduction of Problem~\ref{def:sparse_linsolve} to Problem~\ref{def:shortest_repr_problem}, which allows us to characterize the complexity of the problem of finding sparse representations in terms of the complexity of Problem~\ref{def:sparse_linsolve}.
For the practically most relevant case of $K = \QQ$ we arrive at the following theorem.

\begin{theorem}
  \label{prop:NP}
  The problem of, given $f, f_1,\dots,f_{r} \in \QQ\langle X \rangle$ 
  and $N \in \NN$ (in unary form), deciding whether there exists a cofactor representation of $f$ of weight at most $N$, is \NP-complete.
\end{theorem}
\begin{proof}
  Over $\QQ$, the decision problem associated to Problem~\ref{def:sparse_linsolve} is \NP-complete~\cite[Problem MP5]{Garey-1990-ComputersIntractabilityGuide}.
  Given an input $A,\mathbf{b},N$ to that problem, introduce one variable $x_{i}$ for each row of $A$, interpret each column of $A$ as the polynomial $f_j = \sum_i A_{i,j} x_i$, and the right-hand side as the polynomial  $f = \sum_i \mathbf{b}_i x_i$.
  There is a one-to-one correspondence between solutions with $N$ nonzero entries to the linear system and cofactor representations of $f$ with weight $N$.
  So the problem of finding a representation of weight at most $N$ is also \NP-hard.

  Furthermore, if there exists a representation of weight $\leq N$, then there exists one with degree $\leq \deg(f) + N \max_i \degdiff(f_i)$, which makes it polynomial size in $N$ and the size of the input polynomials.
  The validity of that representation can be verified in polynomial time.
  So the problem is \NP, and therefore \NP-complete.
\end{proof}

\begin{remark}
  The requirement that $N$ be given in unary format is necessary in order to get a bound in terms of the size of the input: unlike in Problem~\ref{def:sparse_linsolve}, the input of Problem~\ref{def:shortest_repr_problem} is not at least of size $N$.
  If $N$ is given in binary format, the size of the input is $\log(N) + {}$the length of a representation of the polynomials, and the latter is independent of $N$. In that case, the decision problem is still \NP-hard but no longer \NP, because the degree bound is not polynomial in $\log(N)$.
\end{remark}

Since Problem~\ref{def:sparse_linsolve} is \NP-complete, Theorem~\ref{prop:NP} implies that there exists a polynomial-time reduction of Problem~\ref{def:shortest_repr_problem} to Problem~\ref{def:sparse_linsolve}.
However, Algorithm~\ref{algo:decidability} is not such a reduction, as it reduces to an instance of Problem~\ref{def:sparse_linsolve} with size $|L| = O(k^{N})$, where $k$ is the number of variables.
For Problem~\ref{def:sparse_linsolve}, the complexity of a brute-force solver looping over all $N$-subsets of $\{1,\dots,n\}$ is $O(n^{N})$.
All in all, the runtime complexity of Algorithm~\ref{algo:decidability}, using such an algorithm for the last step, is $O(k^{N^{2}})$.



In practice, the last step of Algorithm~\ref{algo:decidability} is infeasible for non-trivial examples.
To illustrate this point, we consider the following simple statement about the Moore-Penrose inverse.

\begin{theorem}[{\cite[Ch.~5.7 Fact 11]{hogben2013handbook}}]\label{thm:MP}
Let $A$ be an invertible matrix with inverse $B$.
Then $B$ is the Moore-Penrose inverse of $A$.
\end{theorem}
\begin{proof}
Let $A^{\dagger}$ be the Moore-Penrose inverse of $A$.
So in particular, $AA^{\dagger}A=A$, and hence $B = BAB = BAA^\dagger AB = BAA^\dagger = A^\dagger$.
\end{proof}

\begin{example}
\label{ex:MP}
Theorem~\ref{thm:MP} can be encoded in terms of the ideal membership $b - a^\dagger \in (F)$
with 
\begin{align*}
	F = \{ab-1, ba-1, a a^\dagger a - a, a^\dagger a a^\dagger - a^\dagger, (a^\dagger)^* a^* - a a^\dagger, a^* (a^\dagger)^* - a^\dagger a\}
\end{align*}
in the algebra $\QQ\langle a,a^{*},a^{\dagger},(a^{\dagger})^{*},b \rangle$.

The proof given above is then equivalent to the following cofactor representation of $b - a^\dagger$ certifying the ideal membership:
\begin{equation}\label{eq:cofactor-repr}
	b - a^\dagger = a^\dagger(ab - 1) - b(ab-1) - b(aa^\dagger a -a )b + (ba - 1)a^\dagger ab.
\end{equation}
This cofactor representation consists of $4$ terms.
To see if there exists a representation with $\leq 3$ terms, we can call Algorithm~\ref{algo:decidability} with $N = 3$.
The set $L$ contains polynomials of degree at most $D=7$, and it consists of \num{88672} elements.
This is too large to test all 3-subsets exhaustively.

Using the techniques of Section~\ref{sec:finding-short-proofs}, we will see that a much smaller set of elements is sufficient and by applying the results of
Section~\ref{sec:special-case}, we will be able to verify that \eqref{eq:cofactor-repr} is in fact a sparsest representation of $b - a^\dagger$.
This shows that the proof given above is a shortest proof of the theorem.
\end{example}


\section{Computing sparse representations}
\label{sec:finding-short-proofs}

For the rest of this work, we restrict ourselves to the case $K = \QQ$.
We have seen in the previous section that computing sparsest cofactor representations is equivalent to the \NP-hard problem of finding sparsest solutions of a linear system.
Several methods have been proposed to obtain approximate solutions of the latter~\cite{coifman1992entropy,mallat1993matching,chen2001atomic} by using other measures as proxies for the sparsity of a solution and by minimizing over them.
One of these methods, called \emph{Basis Pursuit}~\cite{chen2001atomic}, uses the $\ell_1$-norm as an approximation for the sparsity of a solution.

In the following, we follow the Basis Pursuit approach and use the $\ell_1$-norm $\lVert\alpha\rVert_1 \coloneqq \sum_{i=1}^d |c_i|$ of $\alpha = \sum_{i=1}^d c_i a_i \varepsilon_{j_i} b_i$ as a surrogate complexity measure of a cofactor representation.
The advantage of this approach is that an $\ell_1$\nobreakdash-minimal solution of a linear system over $\QQ$ can be found efficiently using linear programming.
Additionally, we use the effective description of the syzygy module provided by signature-based algorithms to reduce the size of the linear system that we have to consider.

Based on Corollary~\ref{prop:bound-shortest-repr}, it suffices to consider only cofactor representations up to a degree bound when computing minimal representations.
Here, we, more generally, restrict to representations with signature less than a designated bound $\sigma \in M(\Sigma)$.
Since the module ordering is assumed to be fair, this ensures that we work in a finite dimensional vector space.
If the module ordering is also compatible with the degree, \ie, if $\deg(\alpha) \leq \deg(\beta)$ implies $\alpha \preceq \beta$, this includes all cofactor representations of degree $< \deg(\sigma)$.

So, formally, we seek a minimal element w.r.t.\ $\lVert\cdot\rVert_1$ in the set 
\[
	R(f,\sigma) \coloneqq \left\{\alpha \in R(f) \mid \sig(\alpha) \prec \sigma\right\}
\]
of cofactor representations of $f$ \emph{up to signature $\sigma \in M(\Sigma)$}.
We denote the set of all such $\ell_1$-minimal elements by $R_1(f,\sigma)$. 
Analogously, we let $R_0(f,\sigma)$ be the set of all minimal elements w.r.t.\ $\lVert\cdot\rVert_0$ in $R(f,\sigma)$.

The results in this section rely on the fact that we have some $\alpha \in R(f,\sigma)$.
However, in general, for $\sigma$ too small, the set $R(f,\sigma)$ can be empty, even if $R(f) \neq \emptyset$.
To resolve this issue, we assume that we have a cofactor representation $\alpha \in R(f)$ and that $\sigma$ is chosen so that $\sigma \succ \sig(\alpha)$.
Note that this assumption, in particular, implies that $f \in (f_1,\dots,f_r)$.
Such $\alpha$ can be obtained, for example, by reducing $f$ to zero using a (partial) labeled Gr\"obner basis and keeping track of the reductions.
With this in mind, we assume that $R(f,\sigma) \neq \emptyset$.
  
In the following, we describe Algorithm~\ref{algo:short-rep}, which allows to compute an element in the set $R_{1}(f,\sigma)$.
To this end, we denote by $I^{[\Sigma]}$ the labeled module generated by $f_1,\dots,f_r$ and by $H_\sigma$ a Gr\"obner basis of $\Syz(I^{[\Sigma]})$ up to signature $\sigma$.

The general idea of Algorithm~\ref{algo:short-rep} is still to reduce the problem of computing sparse cofactor representations to computing certain solutions of a linear system.
However, instead of choosing all polynomials $a f_i b$ to form the linear system like Algorithm~\ref{algo:decidability} does,
we use the information provided by $H_\sigma$ to trim this set.
More precisely, we find a finite set of module monomials $B = \{\mu_1,\dots,\mu_d\} \subseteq \MFr$ such that
$R_{i}(f,\sigma)$, $i = 0,1$, has non-empty intersection with the $\QQ$\nobreakdash-vector space generated by $B$ and then only consider the polynomials $\{\overline \mu_1,\dots, \overline \mu_d\}$ to form the linear system.
Furthermore, we now no longer seek a sparsest solution of the resulting system but an $\ell_1$-minimal solution, which can be found with linear programming.

It remains to discuss how to find a suitable basis $B$ and how to translate the problem of finding $\ell_1$-minimal solutions of a linear system into a linear programming problem.

\subsection{Finding a suitable basis $B$}\label{sec:finding-basis}

Algorithm~\ref{algo:decidability} essentially uses the basis $B = \{a \varepsilon_i b \mid a,b \in \<X>, i = 1,\dots,r, \sig(a\varepsilon_i b) \prec \sigma\}$,
which leads to finite dimensional, yet infeasibly large, linear systems.
Using a Gr\"obner basis of $\Syz(I^{[\Sigma]})$ up to signature $\sigma$, we can drastically reduce the dimension of the search space. 
To this end, we extend the notion of rewriting to module elements.

\begin{definition}
Let $\alpha,\gamma \in \Sigma$ and $a,b \in \<X>$ such that $\supp(\alpha) \cap \supp(a\gamma b) \neq \emptyset$.
For every $c \in \QQ$, we say that $\alpha$ \emph{can be rewritten to $\alpha + ca\gamma b$ by $\gamma$}.

Furthermore, we say that $\alpha$ \emph{can be rewritten to $\beta$ by} $H \subseteq \Sigma$ if there are $\beta_0,\dots,\beta_d \in \Sigma$,
$\beta_d = \alpha$, $\beta_0 = \beta$ and $\gamma_1,\dots,\gamma_d \in H$ such that $\beta_{k}$ can be rewritten to $\beta_{k-1}$ by $\gamma_k$ for all $k = 1,\dots,d$.
\end{definition}

With this, we can state a module version of Lemma~\ref{lemma:rewriting-short-rep}.
We note that we state all results in this section for both the $\ell_0$-``norm'' and the $\ell_1$-norm to emphasize that they hold for both complexity measures likewise and that the restriction to $\lVert\cdot\rVert_1$ only comes later for the linear programming.

\begin{lemma}\label{lemma:module-rewriting}
Let $i \in \{0,1\}$.
Furthermore, let $\alpha \in R(f,\sigma)$, $\alpha_i \in R_i(f,\sigma)$, and let $H_\sigma$ be a Gr\"obner basis of $\Syz(I^{[\Sigma]})$ up to signature $\sigma$.
Then $\alpha$ can be rewritten to $\alpha_i$ by $H_\sigma$.
In particular, this rewriting can be done so that the signature of every rewriter $a_j \gamma_j b_j$ is less than $\sigma$.
\end{lemma}

To prove Lemma~\ref{lemma:module-rewriting}, we make use of the fact that the $\ell_0$\nobreakdash-``norm'' and the $\ell_1$-norm are linear for elements of disjoint support.
\begin{lemma}\label{lemma:disjoint}
For $\alpha,\beta \in \Fr$ with $\supp(\alpha) \cap \supp(\beta) = \emptyset$, we have $\lVert \alpha + \beta\rVert_i = \lVert \alpha \rVert_i + \lVert \beta\rVert_i$ for $i = 0,1$.
\end{lemma}

\begin{proof}[Proof of Lemma~\ref{lemma:module-rewriting}]
The difference $\alpha - \alpha_i$ is a syzygy with signature $\prec \sigma$.
Since $H_\sigma$ is a Gr\"obner basis of $\Syz(I^{[\Sigma]})$ up to signature $\sigma$, there exist $d \in \NN$ and $\gamma_j \in H_\sigma$, $c_j \in \QQ$, $a_j,b_j \in \<X>$ such that $\alpha_i = \alpha - \sum_{j = 1}^d c_j a_j \gamma_j b_j$
and $\sig(a_j \gamma_j b_j) \preceq \max\{\sig(\alpha),\sig(\alpha_i)\} \prec \sigma$ for all $j$.
Now, we essentially follow the proof of Lemma~\ref{lemma:rewriting-short-rep} and perform induction on $d$.

The case $d = 0$ is clear.
Assume now that $d > 0$ and that the result is proven for all pairs $(\alpha,\alpha_i)$ such that $\alpha - \alpha_i$ has a representation with $d-1$ terms.
Let $\beta = \sum_{j = 1}^d c_j a_j \gamma_j b_j$.
If $\beta = 0$, we are done since $\alpha = \alpha_i$.
So assume $\beta \neq 0$, which implies $\lVert \beta \rVert_i > 0$.
Then we must have $\supp(\alpha) \cap \supp(\beta) \neq \emptyset$, as otherwise Lemma~\ref{lemma:disjoint} would yield the contradiction
$\lVert \alpha_i \rVert_i = \lVert \alpha - \beta \rVert_i = \lVert \alpha \rVert_i + \lVert \beta \rVert_i > \lVert \alpha \rVert_i \geq \lVert \alpha_i \rVert_i$,
where the last inequality follows from the minimality of $\lVert \alpha_i \rVert_i$.
Thus, we have $\supp(\alpha) \cap \supp(a_j \gamma_j b_j) \neq \emptyset$ for some $1 \leq j \leq d$.
W.l.o.g.\ assume $j = d$.
Hence, $\alpha$ can be rewritten to $\beta_{d-1} = \alpha - c_d a_d \gamma_d b_d$ by $\gamma_d \in H_\sigma$.
Note that $\sig(a_d \gamma_d b_d) \prec \sigma$.
Since $\beta_{d-1} - \alpha_i = \sum_{j=1}^{d-1} c_j a_j \gamma_j b_j$ has a representation with $d-1$ terms, the induction hypothesis implies that $\beta_{d-1}$ can be rewritten to $\alpha_i$ by $H_\sigma$ using only rewriters $a_j \gamma_j b_j$ with signature $\prec \sigma$.
\end{proof}

Lemma~\ref{lemma:module-rewriting} says that any $\alpha \in R(f,\sigma)$ can be rewritten to each element in $R_i(f,\sigma)$, $i = 0,1$, by $H_\sigma$
using only rewriters with signature bounded by $\sigma$.
Consequently, to find a suitable basis $B$, it suffices, starting from some $\alpha$, to only choose those syzygies that can appear in such rewriting sequences.
Finding these elements is a purely combinatorial problem that can be solved without performing any actual rewriting steps.
This leads to Algorithm~\ref{algo:sym-pre}, in which we collect precisely all those relevant syzygies.
Algorithm~\ref{algo:sym-pre} can be considered as an adaptation of the symbolic preprocessing in the F4 algorithm~\cite{faugere1999new}.
In the following, for $V \subseteq \Fr$, let $\supp(V) = \bigcup_{\gamma \in V} \supp(\gamma)$.

\begin{algorithm}
  \SetAlgoLined
  \SetAlgoNoEnd
  \KwIn{$\alpha \in R(f,\sigma)$, $H_\sigma$ a Gr\"obner basis of $\Syz(I\Sig)$ up to signature $\sigma$}
  \KwOut{\mbox{$V\subseteq \Syz(I\Sig)$ such that $R_i(f,\sigma)\subseteq\alpha+\myspan_\QQ(V)$,~$i = 0,1$}}
   $V \leftarrow \emptyset$\;
   $\texttt{todo} \leftarrow \supp(\alpha)$\;
   $\texttt{done} \leftarrow \emptyset$\;
   \While{$\texttt{todo} \neq \emptyset$}{
    select $\mu \in \texttt{todo}$, remove it, and add it to \texttt{done}\;
    $\texttt{new} \leftarrow \{ a\gamma b \mid a,b\in\<X>, \gamma \in H_\sigma, \mu \in \supp(a\gamma b)$, $\sig(a\gamma b) \prec \sigma\}$\;
    $\texttt{todo} \leftarrow \texttt{todo} \cup (\supp(\texttt{new}) \setminus \texttt{done})$\;
    $V \leftarrow V \cup \texttt{new}$\;
  }
  \Return $V$\label{line:return}\;
  \caption{Finding relevant syzygies}
  \label{algo:sym-pre}
\end{algorithm}
\begin{proposition}\label{prop:sym-pre}
Algorithm~\ref{algo:sym-pre} terminates and is correct.
\end{proposition}
\begin{proof}
The conditions on the elements in \texttt{new} ensure that only module monomials smaller than $\sigma$ are inserted into \texttt{todo}.
Furthermore, each monomial is processed at most once. 
Consequently, termination follows from the fact that there are only finitely many monomials smaller than $\sigma$ (recall that $\preceq$ is fair).
Correctness follows from Lemma~\ref{lemma:module-rewriting}.
\end{proof}


Using Algorithm~\ref{algo:sym-pre}, we can set $B = \supp(\alpha) \cup \supp(V)$ as a basis of the search space, where $V$ is the output of the algorithm given $\alpha$ and $H_\sigma$ as input.
In many cases, this set is small enough to reasonably work with. 

\subsection{Detecting redundant syzygies}\label{sec:pruning}

As an optional step, we can remove redundant elements from $V$ before forming the basis $B$ in order to obtain a smaller basis, and thus, a smaller linear program to solve.
More precisely, since we only want to compute one element in $R_i(f,\sigma)$, $i = 0,1$, we can remove syzygies as long as we can ensure that there remains at least one rewriting sequence from $\alpha$ to at least one element in $R_i(f,\sigma)$.
We mention two basic techniques that turned out useful in practice.

The first technique allows to remove syzygies from $V$ that consist mostly of terms that appear in no other element.
Such syzygies cannot lead to simpler representations.
To make this statement precise, for $W \subseteq V$ and $\beta = \sum_j  c_j \mu_j \in W$ with $c_j \in \QQ$, $\mu_j \in \MFr$, we denote
\begin{align*}
	\beta_U &\coloneqq \textstyle{\sum_j} c_j \mu_j  \text{ with }j \text{ such that } \mu_j \notin \supp((V \cup \{\alpha\}) \setminus \{\beta\}),\\
	\beta_V &\coloneqq \textstyle{\sum_j} c_j \mu_j  \text{ with }j \text{ such that } \mu_j \in \supp((V \cup \{\alpha\}) \setminus W).
\end{align*}
Intuitively, the element $\beta_U$ contains all those terms of $\beta$ that are unique to $\beta$ and that appear in no other element of $V \cup \{\alpha\}$,
and $\beta_V$ contains those terms that appear in $\beta$ as well as in elements outside $W$.

\begin{proposition}\label{prop:redundant}
Let $i \in \{0,1\}$.
Furthermore, let $\alpha \in R(f,\sigma)$ and $V \subseteq \Syz(I\Sig)$ such that 
\[
	(\alpha + \myspan_\QQ(V)) \cap R_i(f,\sigma) \neq \emptyset.
\]
If $W \subseteq V$ satisfies $\lVert \beta_V \rVert_i \leq \lVert \beta_U \rVert_i$ for all $\beta \in W$, then
\[
	(\alpha + \myspan_\QQ(V \setminus W)) \cap R_i(f,\sigma) \neq \emptyset.
\]
\end{proposition}

Proposition~\ref{prop:redundant} provides a sufficient condition for a subset $W \subseteq V$ to be redundant. 
In order to prove this, we need the following two lemmas.
The first one states that the required property of $W$ extends to the whole linear span.
To this end, we extend the definition of $\beta_U$ and $\beta_V$ to elements $\beta = \sum_j b_j \beta_j \in \myspan_\QQ(W)$,
where $b_j \in \QQ$ and $\beta_j \in W$, by 
$\beta_U \coloneqq \sum_j b_j \beta_{j,U}$ and $\beta_V \coloneqq \sum_j b_j \beta_{j,V}$.
Note that, in general, these definitions depend on the representation of $\beta$ in terms of the elements in $W$;
different linear combinations of the same element can yield different definitions of $\beta_{U}$ and $\beta_{V}$.
Therefore, to obtain an unambiguous definition, we assume that, for every element $\beta \in \myspan_\QQ(W)$, one particular 
representation in terms of $W$ has been fixed, and this is the representation that we use to compute $\beta_{U}$ and $\beta_{V}$.
We note that, for all our arguments, the particular choice of the representation does not matter.
It is only important that, for each $\beta$, the elements $\beta_{U}$ and $\beta_{V}$ are computed with respect to the same representation.

\begin{lemma}\label{lemma:redundant1}
Let $V,W$ be as in Proposition~\ref{prop:redundant}.
If $\beta \in \myspan_\QQ(W)$, then $\lVert \beta_V \rVert_i \leq \lVert \beta_U \rVert_i$.
\end{lemma}
\begin{proof}
Write $\beta = \sum_j b_j \beta_j$ with nonzero $b_j \in \QQ$ and $\beta_j \in W$. 
By assumption  $\lVert \beta_{j,V} \rVert_i \leq \lVert \beta_{j,U} \rVert_i$ for all $j$.
Furthermore, all $\beta_{j,U}$ have pairwise disjoint supports as they consist of the monomials that are unique to each $\beta_j$.
So Lemma~\ref{lemma:disjoint} implies that $\lVert \cdot \rVert_i$ is linear on linear combinations of the $\beta_{j,U}$.
Using this and the triangular inequality, we get with $c_j = 1$ if $i = 0$ and $c_j = |b_j|$ if $i = 1$:
\begin{align*}
	\lVert \beta_V \rVert_i &\leq \textstyle{\sum_j} \lVert b_j \beta_{j,V} \rVert_i = \textstyle{\sum_j} c_j \lVert \beta_{j,V} \rVert_i \\
	&\leq \textstyle{\sum_j} c_j \lVert \beta_{j,U} \rVert_i = \textstyle{\sum_j} \lVert b_j \beta_{j,U} \rVert_i = \lVert \beta_U \rVert_i.\qedhere
\end{align*}
\end{proof}

The second lemma provides a lower bound on the norm of sums $\gamma + \beta \in \alpha + \myspan_\QQ(W)$.

\begin{lemma}\label{lemma:redundant2}
Let $\alpha,V,W$ be as in Proposition~\ref{prop:redundant}.
If $\beta \in \myspan_\QQ(W)$ and $\gamma \in \alpha + \myspan_\QQ(V \setminus W)$, then $\lVert \gamma + \beta \rVert_i \geq \lVert\gamma\rVert_i - \lVert \beta_V\rVert_i + \lVert\beta_U\rVert_i$.
\end{lemma}

\begin{proof}
Let $\beta' = \beta - (\beta_U + \beta_V)$.
By definition, $\beta_U$ and $\beta'$ have pairwise different supports.
Furthermore, $\gamma + \beta_V$ does not share a monomial with $\beta_U$ and $\beta'$
as $\supp(\gamma + \beta_V) \subseteq \supp\left((V \cup \{\alpha\}) \setminus W\right)$ and all monomials of $\beta$ that lie in this set are collected in $\beta_V$.
Therefore, Lemma~\ref{lemma:disjoint} and the inverse triangle inequality imply
\begin{align*}
	\lVert\gamma + \beta\rVert_i &= \lVert \gamma + \beta_V\rVert_{i} + \lVert\beta_U\rVert_i +  \lVert\beta'\rVert_i \\
	&\geq \lVert \gamma + \beta_V\rVert_{i} + \lVert\beta_U\rVert_i \geq \lVert \gamma\rVert_ i - \lVert\beta_V\rVert_i + \lVert\beta_U\rVert_i.\qedhere
\end{align*}
\end{proof}

\begin{proof}[Proof of Proposition~\ref{prop:redundant}]
We claim that removing, if present, elements from $W$ from a representation $\delta \in \alpha + \myspan_\QQ(V)$ cannot increase the norm.
This implies the assertion of the proposition.
To prove our claim, write $\delta$ as $\delta = \gamma + \beta$ 
with $\gamma \in \alpha + \myspan_\QQ(V \setminus W)$ and $\beta \in \myspan_\QQ(W)$.
Now, Lemma~\ref{lemma:redundant1} and~\ref{lemma:redundant2}, show that
$\lVert\delta \rVert_i = \lVert\gamma + \beta \rVert_i \geq  \lVert\gamma \rVert_i - \lVert\beta_{V} \rVert_i + \lVert\beta_{U} \rVert_i \geq \lVert\gamma\rVert_i$.
\end{proof}

The redundancy test provided by Proposition~\ref{prop:redundant} is computationally fairly cheap to check for a given set $W \subseteq V$.
However, finding suitable candidates for $W$ is not so trivial.
In our implementation, we test all singletons $\{\beta\} \subseteq V$ and all subsets
$\{\beta, \gamma\} \subseteq V$ where $\supp(\beta) \cap \supp(\gamma) \neq \emptyset$ and
where $\beta$ and $\gamma$ consist, to at least a third, of unique monomials that appear in no other element in $V$.
Empirically, this provided the best trade-off between efficiency in applying the criterion and the effect it had on pruning $V$.

The second method does not directly allow to detect redundant elements in $V$.
Instead it can be considered as an auxiliary technique that can cause additional applications of Proposition~\ref{prop:redundant}.
The idea is to replace elements in $V$ by linear combinations so that the number of occurrences of certain monomials is reduced. 
In particular, by exploiting the fact that
\begin{equation}\label{eq:span}
	\myspan_\QQ(V \cup \{ \alpha - \beta, \gamma + \beta \}) = \myspan_\QQ(V \cup \{ \alpha - \beta, \alpha + \gamma \}),
\end{equation}
we can reduce the number of occurrences of $\beta$ at the cost of increasing the occurrences of $\alpha$.

In our implementation, we apply this technique to all binomial syzygies $\mu - \sigma \in V$.
After removing all occurrences of $\sigma$, Proposition~\ref{prop:redundant} allows to delete the binomial syzygy from $V$.
Additionally, we apply~\eqref{eq:span} randomly to elements $\alpha - \beta$ where $\lVert \beta \rVert_i > c \lVert \alpha \rVert_i$ for fixed $c > 1$. 
Often, this process triggers further invocations of Proposition~\ref{prop:redundant} to remove elements from $V$.
Table~\ref{table:comparison} shows the efficiency of the two methods presented in this section.

%

\subsection{Translation into linear program}\label{sec:linear-system}
Once we have obtained a reasonable basis of module monomials $B = \{\mu_1,\dots,\mu_d\}$ such that the $\QQ$-vector space generated by $B$ has non-empty intersection with $R_i(f,\sigma)$, $i = 0,1$,
we can set up a linear system $A \mathbf{y} = \mathbf{b}$, where $A$ is the matrix of size $s \times d$, with $s = | \bigcup_j \supp(\overline \mu_j) |$, whose $j$th column contains the coefficients of $\overline \mu_j$, associating to each row of $A$ a monomial $m \in \bigcup_j \supp(\overline \mu_j)$.
Similarly, $\mathbf{b}$ is a vector of size $s$ containing the coefficients of $f$.
The matrix $A$ bears resemblance to the matrices appearing in Gr\"obner basis computations such as the F4 algorithm, aside from two main differences.
In Gr\"obner basis computations, polynomials are encoded as the rows of a matrix and the columns have to be ordered w.r.t.\ the (polynomial) monomial ordering.
In our approach, polynomials are encoded as the columns and the order of the rows is irrelevant.

Every solution $\mathbf{y}$ of $A\mathbf{y} = \mathbf{b}$ corresponds to a cofactor representation of $f$ with support in $B = \{\mu_1,\dots,\mu_d\} \subseteq \MFr$.
To see this, we denote by $A_{i,j}$ the entries of $A$ and by $\mathbf{b}_{i}$ and $\mathbf{y}_{j}$ the coordinates of $\mathbf{b}$ and $\mathbf{y}$ respectively.
Furthermore, let $m_{i} \in \bigcup_j \supp(\overline \mu_j)$ be the monomial that is associated to the $i$th row of $A$ and the $i$th coordinate of $\mathbf{b}$.
Then we can write $f = \sum_{i=1}^{s} \mathbf{b}_{i} m_{i}$ and $\overline \mu_{j} = \sum_{i=1}^{s} A_{i,j} m_{i}$, for $j = 1,\dots,d$, and we see that
\begin{align*}
	f = \sum_{i=1}^{s} \mathbf{b}_{i} m_{i} = \sum_{i=1}^{s} \left(\sum_{j=1}^{d} A_{i,j} \mathbf{y}_{j}\right) m_{i} = \sum_{j=1}^{d} \mathbf{y}_{j} \left(\sum_{i=1}^{s} A_{i,j} m_{i}\right) = \sum_{j=1}^{d} \mathbf{y}_{j} \overline \mu_{j},
\end{align*}
showing that $\sum_{j=1}^{d} \mathbf{y}_{j} \mu_{j} \in R(f,\sigma)$ is a cofactor representation of $f$.

Moreover, every $\ell_i$-minimal solution of $A\mathbf{y} = \mathbf{b}$ corresponds to an element in $R_i(f,\sigma)$. 
As noted before, computing $\ell_0$-minimal, i.e., sparsest, solutions is \NP-hard.
Therefore, we restrict ourselves to the case $i = 1$ and consider the problem
\begin{equation*}
\label{eq:P1}
	(P_1): \quad \min_\mathbf{y} \, \lVert \mathbf{y} \rVert_1 \quad \text{ subject to }\quad  A\mathbf{y} = \mathbf{b},
\end{equation*}
where $\lVert\mathbf{x}\rVert_1 = \sum_j |\mathbf{x}_j |$.
%
It is well-known that \hyperref[eq:P1]{$(P_1)$} can be recast as a linear program, see e.g.~\cite[Sec.~3.1]{chen2001atomic}.
A linear program (in standard form)~\cite{schrijver1998theory} is an optimization problem for $\mathbf{v} \in \QQ^t$ of the form
\begin{equation*}
\label{eq:LP}
	(LP): \quad \min_\mathbf{v} \, \mathbf{c}^T \mathbf{v} \quad \text{ subject to }\quad  U\mathbf{v} = \mathbf{w}, \quad \mathbf{v} \geq 0,
\end{equation*}
where $\mathbf{v} \geq 0$ is to be understood component-wise.
The problem \hyperref[eq:P1]{$(P_1)$} can be equivalently formulated as a linear program by setting
\[
	t = 2d, \quad \mathbf{c}^T = (1,\dots,1), \quad U = (A \mid - A), \quad \mathbf{v} = \begin{pmatrix} \mathbf{p} \\\mathbf{q} \end{pmatrix}, \quad \mathbf{w} = \mathbf{b},
\]
with vectors $\mathbf{p}, \mathbf{q} \in \QQ^d$.
This linear program can then be solved efficiently using the simplex algorithm~\cite{dantzig1951maximization} or interior-point methods~\cite{potra2000interior}
and a solution $\mathbf{y}$ of \hyperref[eq:P1]{$(P_1)$} is given by $\mathbf{y} = \mathbf{p} - \mathbf{q}$.

\subsection{Putting everything together}

Finally, we combine the results of the previous sections to form Algorithm~\ref{algo:short-rep} for computing an element in $R_1(f,\sigma)$.
In the algorithm, $I^{[\Sigma]}$ denotes the labeled module generated by $f_1,\dots,f_r$.

\begin{algorithm}
  \SetAlgoLined
  \KwIn{$f_1,\dots,f_r \in \QQ\<X>$, $f \in (f_1,\dots,f_r)$, $\sigma \in \MFr$, $\alpha \in R(f,\sigma)$}
  \KwOut{an element in $R_1(f,\sigma)$}
  
  $H_\sigma \leftarrow $ Gr\"obner basis of $\Syz(I^{[\Sigma]})$ up to signature $\sigma$\;
  
   $V \leftarrow$ apply Algorithm~\ref{algo:sym-pre} to $\alpha$ and $H_\sigma$\;
   $V \leftarrow$ prune $V$ using the techniques from Section~\ref{sec:pruning}\;
  $\{\mu_1, \dots,\mu_d\} \leftarrow \supp(V \cup \{\alpha\})$\;
  
  $A \leftarrow $ matrix with columns containing the coefficients of $\overline \mu_1,\dots,\overline \mu_d$\;
  $\mathbf{b} \leftarrow$ vector containing the coefficients of $f$\;
  $\mathbf{v} \leftarrow$ solution of the linear program \hyperref[eq:LP]{$(LP$)}
 with $\mathbf{c}^T = (1,\dots,1),\quad U = (A \mid - A), \quad \mathbf{v} = \begin{pmatrix} \mathbf{p} \\\mathbf{q} \end{pmatrix}, \quad \mathbf{w} = \mathbf{b}$\;\label{line:linear-program}

 \Return $\sum_{i=1}^d (\mathbf{p}_i - \mathbf{q}_i) \mu_i$\;
 
  \caption{$\ell_1$-minimal cofactor representation}
  \label{algo:short-rep}
\end{algorithm}

\begin{theorem}
Algorithm~\ref{algo:short-rep} terminates and is correct.
\end{theorem}

\begin{proof}
Termination follows from the fact that $H_\sigma$ can be computed in finite time by Theorem~\ref{thm:syz-basis}, and from Proposition~\ref{prop:sym-pre}.
Correctness follows from the discussions in Section~\ref{sec:finding-basis},~\ref{sec:pruning} and~\ref{sec:linear-system}.
\end{proof}

\begin{remark}
  The complexity of Algorithm~\ref{algo:short-rep}, in the worst case, is polynomial in the size of the matrix built, or singly exponential in the size of the input.
  This is the same complexity as one would achieve in Algorithm~\ref{algo:decidability} by using $\ell_{1}$-relaxation and an LP-solver instead of the last step.

  In particular, it is not clear how much the use of signatures and the pruning techniques can reduce the size of the matrix in the worst case, so they do not change the worst-case complexity bounds.
  However, they have a large impact on the complexity in practice, as the experimental results in Section~\ref{sec:experimental-results} demonstrate.
\end{remark}


\begin{remark}\label{remark:weights}
Algorithm~\ref{algo:short-rep} weighs each monomial $\mu_i$ equally by a weight of $1$.
It is also possible to weigh the monomials differently by changing the vector $\mathbf{c}$ so that $\mathbf{c}_i$ encodes the weight of $\mu_i$.
This allows, for example, to weigh monomials by their degree, yielding representations that prefer monomials with small degree.
In this case, the output of the algorithm is no longer guaranteed to be in $R_1(f,\sigma)$, as witnessed by the example \texttt{ROL-3} discussed in Section~\ref{sec:experimental-results}.
\end{remark}

\subsection{Special case: totally unimodular matrices}\label{sec:special-case}

In general, the output of Algorithm~\ref{algo:short-rep} need not be a sparsest representation of $f$ up to signature $\sigma$, \ie, it need not be an element in $R_0(f,\sigma)$.
In this section, we discuss a special case when this is indeed true.
To this end, we consider the linear system $A\mathbf{y} = \mathbf{b}$ constructed in Algorithm~\ref{algo:short-rep}.
We are interested in situations where the augmented matrix $(A \mid \mathbf{b})$ is \emph{totally unimodular} as defined below.

\begin{definition}
A matrix $T \in \{-1,0,1\}^{m \times n}$ is called \emph{totally unimodular} if every square submatrix of $T$ has determinant $0$ or $\pm 1$.
\end{definition}

\begin{theorem}\label{thm:totally-unimodular}
Let $A$ and $\mathbf{b}$ be as constructed in Algorithm~\ref{algo:short-rep}.
If the augmented matrix $(A \mid \mathbf{b})$ is totally unimodular, then 
the output of Algorithm~\ref{algo:short-rep} is an element in $R_0(f,\sigma)$.
\end{theorem}

In order to prove the theorem, we take a closer look at the coefficients of the sparsest and $\ell_1$-minimal solutions of $A\mathbf{y} = \mathbf{b}$.
It is well-known that totally unimodular coefficient matrices and integer right-hand sides yield integer optima for linear programs~\cite[Cor.~19.1a]{schrijver1998theory}.
The following lemma, extends this statement under slightly stricter assumptions.

\begin{lemma}\label{lemma:sol-P01}
Let the augmented matrix $(A\mid \mathbf{b})$ be totally unimodular.
If $A\mathbf{y} = \mathbf{b}$ is solvable, then any sparsest or $\ell_1$-minimal solution $\mathbf{y}$ satisfies $\mathbf{y} \in \{-1,0,1\}^d$.
\end{lemma}

\begin{proof}
Since removing linearly dependent rows does not change the solution set of a solvable system, we can assume that $A$ has full row rank $s = \rank(A)$.

\emph{Sparsest solution.}
The columns of $A$ corresponding to the nonzero entries of $\mathbf{y}$ have to be linearly independent
(otherwise there would exist a sparser solution).
We can extend them by further columns of $A$ to obtain an invertible $s \times s$ matrix $A'$.
Then $A' \mathbf{y}' = \mathbf{b}$, where $\mathbf{y}'$ contains those coordinates of $\mathbf{y}$ that correspond to the columns of $A$ that are in $A'$. 
By assumption $\det(A') = \pm 1$.
Furthermore, the matrix $A'_i$ obtained by replacing the $i$th column of $A'$ by $\mathbf{b}$ is -- up to permutation of columns -- a submatrix of $(A\mid \mathbf{b})$.
Consequently, $\det(A'_i) \in \{-1,0,1\}$ and applying Cramer's rule shows $\mathbf{y}_i' = \frac{\det(A'_i)}{\det(A')} \in \{-1,0,1\}$.
Then the result follows since any coordinate of $\mathbf{y}$ which does not appear in $\mathbf{y'}$ has to be zero.

\emph{$\ell_1$-minimal solution.}
We consider the equivalent linear program \hyperref[eq:LP]{$(LP$)} and note that $\mathbf{y} \in \{-1,0,1\}^d$ if and only if $\mathbf{v}\in \{0,1\}^{2d}$.
If $\mathbf{v}$ is a solution of \hyperref[eq:LP]{$(LP$)}, then it has to be a basic feasible solution.
This means $\lVert \mathbf{v} \rVert_0 \leq s$ and that the columns of $U$ that correspond to the nonzero coordinates of $\mathbf{v}$ can be extended to an invertible $s \times s$ submatrix $U'$ of $U$.
Since $(U \mid \mathbf{b}) = (A \mid - A \mid \mathbf{b})$ is totally unimodular, the same arguments as in the other case show that $\mathbf{v} \in \{-1,0,1\}^{2d}$,
and the statement follows from the non-negativity constraint of \hyperref[eq:LP]{$(LP$)}.
\end{proof}

Using this lemma, we can now prove Theorem~\ref{thm:totally-unimodular}.

\begin{proof}[Proof of Theorem~\ref{thm:totally-unimodular}]
By construction, the system $A\mathbf{y} = \mathbf{b}$ has a solution.
For $i =0,1$, let $\alpha_i$ be the module element corresponding to an $\ell_i$-minimal solution of the system.
Note that, again by construction, $\alpha_i \in R_i(f,\sigma)$.
By Lemma~\ref{lemma:sol-P01}, $\alpha_i$ contains only nonzero coefficients $\pm 1$, which implies that $\lVert\alpha_i\rVert_0 = \lVert\alpha_i\rVert_1$ for $i = 0,1$, and the result follows.
\end{proof}

In most applications, all polynomials involved are of the form $a - b$ with $a,b\in \<X> \cup \{0\}$ encoding identities of operators of the form $A = B$.
Such polynomials are called \emph{pure difference binomials}.
The following corollary of Theorem~\ref{thm:totally-unimodular} ensures that Algorithm~\ref{algo:short-rep} computes a sparsest representation up to signature $\sigma$ provided that the input polynomials are pure difference binomials.

\begin{corollary}\label{lemma:totally-unimodular}
Let $f, f_1,\dots,f_r\in \QQ\<X>$ be pure difference binomials, $\sigma \in \MFr$, and $\alpha \in R(f,\sigma)$.
Given these elements as input, the output of Algorithm~\ref{algo:short-rep} is an element in $R_0(f,\sigma)$.
\end{corollary}

\begin{proof}
Let $A$, $\mathbf{b}$ as constructed in Algorithm~\ref{algo:short-rep}.
By assumption on $f,f_1,\dots,f_r$, each column of $(A\mid\mathbf{b})$ contains at most one entry $+1$ and at most one entry $-1$ with all other entries being $0$.
Each square submatrix $U$ of $(A\mid\mathbf{b})$ either contains a zero column (then $U$ is singular), 
a column with one nonzero entry (then expansion of $\det(U)$ along this column yields inductively $\det(U) \in \{-1,0,1\}$),
or each column of $U$ contains exactly one entry $+1$ and one entry $-1$ (then  $\mathbf{1}^T U = 0$ showing that $U$ is singular).
Thus, $(A\mid\mathbf{b})$ is totally unimodular and the result follows from Theorem~\ref{thm:totally-unimodular}.
\end{proof}


\begin{example}
We revisit Example~\ref{ex:MP}.
All polynomials that appear in this example are pure difference binomials.
Hence, Corollary~\ref{lemma:totally-unimodular} implies that Algorithm~\ref{algo:short-rep} yields a sparsest cofactor representation up to the used signature bound $\sigma$.
In particular, if a degree-compatible module ordering is used and $\sigma$ is chosen so that $\deg(\sigma) > 7$, then by Corollary~\ref{prop:bound-shortest-repr} the computed representation is a sparsest one (independent of any bound).

Applying Algorithm~\ref{algo:short-rep} to Example~\ref{ex:MP}, with the cofactor representation given in~\eqref{eq:cofactor-repr} and a suitable signature bound $\sigma$,
yields again~\eqref{eq:cofactor-repr}, showing that this is a sparsest cofactor representation.
The basis used to form the linear system only consists of \num{300} elements, compared to the \num{88672} that Algorithm~\ref{algo:decidability} would need.
\end{example}
  
\section{Experimental results}
\label{sec:experimental-results}

We have written a prototype implementation of Algorithm~\ref{algo:short-rep} for \textsc{SageMath}\footnote{\label{foot}Available at \url{https://clemenshofstadler.com/software/}} using
our package \texttt{signature\_gb}\textsuperscript{\ref{foot}} for the signature-based computations and
the \textsc{IBM ILOG CPLEX} optimization studio~\cite{cplex} for the linear programming.

In Table~\ref{table:comparison}, we compare the weight of cofactor representations computed by Algorithm~\ref{algo:short-rep} to those found by other approaches.
In particular, we compare our algorithm to tracing standard Gr\"obner basis computations and reductions, and to tracing reductions done with a signature Gr\"obner basis.

As benchmark examples, we recover short proofs of the following (recent) results in operator theory on the Moore-Penrose inverse. 
\begin{itemize}
	\item \texttt{SVD} encodes~\cite[Ch.~5.7 Fact~4]{hogben2013handbook}, which
	provides a formula for the Moore-Penrose inverse of a matrix in terms of the matrix's singular value decomposition.
	\item \texttt{ROL} encodes the implication $(2) \Rightarrow (1)$ in~\cite[Thm.~3]{koliha2007moore}, which
	provides a sufficient condition for the identity $(AB)^\dagger = B^\dagger A^\dagger$ to hold, where $X^\dagger$ is the Moore-Penrose inverse of an element in a ring with involution.
	\item \texttt{ROL-$n$} encodes the implication $(n) \Rightarrow (1)$ in~\cite[Thm.~2.1]{DD10}.
	This family provides several sufficient conditions for the identity $(AB)^\dagger = B^\dagger A^\dagger$ to hold, where $X^\dagger$ is the Moore-Penrose inverse of a bounded operator on Hilbert spaces.
	\item \texttt{Hartwig-$n$} encodes the implication $(n) \Rightarrow (1)$ in \cite[Thm.~2.3]{cvetkovic2021algebraic}.
	This family provides several sufficient conditions for the identity $(ABC)^\dagger = C^\dagger B^\dagger A^\dagger$ to hold, where $X^\dagger$ is the Moore-Penrose inverse of an element in a ring with involution.
\item \texttt{Ker} encodes part of~\cite[Thm.~1]{RP87},
which characterizes the existence of Moore-Penrose inverses in additive categories with involution in terms of kernels of morphisms.
\item \texttt{SMW} encodes~\cite[Thm.~2.1]{Den11}, which generalizes the Sherman–Morrison–Woodbury formula in terms of the Moore-Penrose inverse.
\item \texttt{Sum} encodes~\cite[Lem.~1]{Yua08}, which provides a sufficient condition for the identity $(A + B)^\dagger = A^\dagger +  B^\dagger$ to hold, where $X^\dagger$ is the Moore-Penrose inverse of an element in a $C^{*}$-algebra.
\end{itemize}

These operator statements are translated into polynomial ideal membership analogous to how Theorem~\ref{thm:MP} and
its proof are translated in Example~\ref{ex:MP}. 
In particular, each identity of operators is translated into a noncommutative polynomial by uniformly replacing each basic operator by a unique noncommutative indeterminate and by forming the difference of the left- and right-hand side of the identity.
With this translation, the assumptions of the operator statement form the generators of the ideal and the claimed identity becomes the polynomial whose ideal membership has to be verified.
We refer to~\cite{raab2021formal} and~\cite[Sec.~5.1]{hofstadler-phd} for a more detailed description of the translation and a theoretical justification of this approach.

For all examples, $\leq_{\textup{deglex}}$ is used in combination with the degree-compatible ordering $\preceq_{\textup{DoPoT}}$ for the signature-based computations. 

The first columns of Table~\ref{table:comparison} contain information about the ideals that arise when translating the operator statements.
In particular, we list the number of generators of each ideal and their maximal degree. 
Moreover, in the column for Algorithm~\ref{algo:short-rep}, we provide information on the used signature bound.
A value $n$ in this column indicates that we consider only cofactor representations of degree $< n$.
The degree bounds were chosen so that the computation would finish for the larger examples \texttt{Hartwig-$n$} and \texttt{ROL-8} within about 90 minutes on a regular laptop
and for the remaining smaller examples within a few minutes.
We note that these degree bounds are strictly smaller than those that Corollary~\ref{prop:bound-shortest-repr} yields, but the latter were computationally infeasible.
Nevertheless, Table~\ref{table:comparison} shows that Algorithm~\ref{algo:short-rep} still allows to find sparser representations for all considered examples.

Apart from the last three (\texttt{Ker}, \texttt{SMW}, \texttt{Sum}), all benchmark examples only consist of pure difference binomials.
For those, Corollary~\ref{lemma:totally-unimodular} implies that the representations computed by Algorithm~\ref{algo:short-rep} are the sparsest up to the respective degree bounds.
For the remaining examples, the algorithm can still be used to find $\ell_{1}$-minimal representations, which are heuristically also sparse, but without guarantee that they are the sparsest.

\begin{table}
  \footnotesize
  \centering
\begin{tabular}{c@{\hskip 7pt}S@{\hskip 3pt}S@{\hskip 3pt}|c@{\hskip 5pt}c@{\hskip -1pt}c@{\hskip 2pt}|S@{\hskip 7pt}S@{\hskip 3pt}S} 
 \toprule
Example & {\#gens} & {deg} & GB & SigGB  & {
    \begin{tabular}[c]{c}
      Algo.~\ref{algo:short-rep} \\ (bound)
    \end{tabular}
    } & {w/o pruning} & {w/ pruning} & {ratio $\neq 0$}\\
 \midrule
\texttt{SVD} &32 &3& 51 & 39 & 25\, (10) &  \qtyproduct{127 x 397}{k} & \qtyproduct{117 x 326}{k} & 0.82 \\
\texttt{ROL} &28 &5& 80 & 39 & 30\, (12)  & \qtyproduct{22 x 102}{k} & \qtyproduct{22 x 55}{k} & 0.54  \\
\texttt{ROL-2} &28 &5& 20 & 21 & 15\, (12) & \qtyproduct{24 x 107}{k} & \qtyproduct{23 x 59}{k} & 0.55 \\
 \texttt{ROL-3} &28 &5& 49 & 44 & 31\, (12) & \qtyproduct{19 x 87}{k} & \qtyproduct{18 x 46}{k} & 0.52 \\
  \texttt{ROL-4} &28&5& 59 & 46 & 33\, (12) & \qtyproduct{68 x 236}{k} & \qtyproduct{64 x 136}{k} & 0.57 \\
 \texttt{ROL-5} & 28 &5& 28 & 30 & 22\, (12)  & \qtyproduct{33 x 134}{k} & \qtyproduct{31 x 79}{k} & 0.59  \\
 \texttt{ROL-6} & 28&5& 39 & 39 & 30\, (12) & \qtyproduct{22 x 99}{k} & \qtyproduct{21 x 54}{k} & 0.54 \\
 \texttt{ROL-7} & 40 &9& 85 & 23 & 17\, (12) & \qtyproduct{18 x 86}{k} & \qtyproduct{17 x 45}{k} & 0.52 \\
 \texttt{ROL-8} &44&7& 241 & 19 & 17\, (12) & \qtyproduct{258 x 965}{k} & \qtyproduct{242 x 548}{k} & 0.57 \\
 \texttt{Hartwig-4} &23&15& 316 & 54 & 46\, (18) &  \qtyproduct{353 x 1756}{k} & \qtyproduct{334 x 1340}{k} & 0.76 \\
  \texttt{Hartwig-5} &26&15& 99 & 43 & 35\, (17)  &\qtyproduct{407 x 1654}{k} & \qtyproduct{392 x 1305}{k} & 0.79 \\
  \texttt{Hartwig-6} &24&15&86 & 33 &  29\, (17) & \qtyproduct{218 x 967}{k} & \qtyproduct{215 x 771}{k} & 0.80 \\
     \hline
     \texttt{Ker} &12 & 3 &49 & 34 &  23\, (12) & \qtyproduct{50 x 142}{k} & \qtyproduct{50 x 129}{k} & 0.90 \\
    \texttt{SMW} &36& 7 & 63 & 42 & 39\, (12) & \qtyproduct{44 x 114}{k} & \qtyproduct{42 x 91}{k} & 0.80 \\
     \texttt{Sum} &20& 3 &313 & 178 &  85\, \phantom{1}(9) & \qtyproduct{11 x 18}{k} & \qtyproduct{10 x 16}{k} & 0.92 \\     \bottomrule
\end{tabular}
\caption{Comparison of weights of cofactor representations computed by standard Gr\"obner bases (GB), by signature Gr\"obner bases (SigGB), and by Algorithm~\ref{algo:short-rep} (Algo.~\ref{algo:short-rep}).
  Size comparison of the coefficient matrix $A$ (rounded to thousands) in Algorithm~\ref{algo:short-rep} with and without applying the pruning techniques from Sec.~\ref{sec:pruning}.
} 
\label{table:comparison}
\end{table}



We also tested an adapted version of Algorithm~\ref{algo:short-rep} as described in Remark~\ref{remark:weights} that minimizes the total number of symbols appearing in a cofactor representation.
For most benchmark examples, the thereby computed representations have the same (minimal) weight as those found with the standard version of the algorithm, but the total number of symbols decreases by up to 15\%.
Only for \texttt{ROL-$3$} does the weight increase by one, while the number of symbols decreases from 196 to 172.
We note that this representation is no longer $\ell_{1}$-minimal.

In the last columns of Table~\ref{table:comparison}, we compare the size of the matrix $A$ constructed in Algorithm~\ref{algo:short-rep}
with and without applying the pruning techniques discussed in Section~\ref{sec:pruning}.
We also list the ratio between the number of nonzero entries in the pruned matrix and the number of nonzero entries in the original matrix.
As the table shows, in some examples the size of the resulting linear system can be reduced drastically, cutting the number of nonzero entries almost in half.

We note that the obtained matrices can still be very large, even after applying the pruning techniques, however they are typically extremely sparse,
which allows to work with them reasonably well.
In particular, if the input only consists of pure difference binomials, then each column contains at most two nonzero entries.
Thus, for example, the largest matrix appearing in our benchmark examples, the one for \texttt{Hartwig-5}, only contains a total of about $2.6$ million nonzero entries, corresponding to a density of nonzero entries of roughly \SI{0.0005}\%.

\section*{Acknowledgements}
C.~Hofstadler was supported by the Austrian Science Fund (FWF) grant P~32301.
T.~Verron was supported by the Austrian Science Fund (FWF) grant P~34872.
We thank Georg Regensburger for valuable remarks which helped to improve the presentation of this work a lot.
We are also grateful to the anonymous referees for their helpful comments and for pointing us to relevant literature.

\bibliographystyle{alpha}
\bibliography{main} 

\newcommand{\etalchar}[1]{$^{#1}$}
\begin{thebibliography}{CIHHP{\etalchar{+}}21}

\bibitem[BHR23]{bernauer2023automatise}
Klara Bernauer, Clemens Hofstadler, and Georg Regensburger.
\newblock {How to Automatise Proofs of Operator Statements: Moore-Penrose
  Inverse; A Case Study}.
\newblock In {\em Proceedings of CASC 2023}, pages 39--68, 2023.

\bibitem[BR22a]{bulatov2022ideal}
Andrei~A. Bulatov and Akbar Rafiey.
\newblock The ideal membership problem and abelian groups.
\newblock In {\em Proceedings of {STACS} 2022}, volume 219 of {\em LIPIcs},
  pages 18:1--18:16, 2022.

\bibitem[BR22b]{bulatov2022complexity}
Andrei~A. Bulatov and Akbar Rafiey.
\newblock {On the Complexity of CSP-Based Ideal Membership Problems}.
\newblock In {\em Proceedings of STOC 2022}, pages 436--449, 2022.

\bibitem[CDS01]{chen2001atomic}
Scott~Shaobing Chen, David~L. Donoho, and Michael~A. Saunders.
\newblock {Atomic Decomposition by Basis Pursuit}.
\newblock {\em {SIAM Rev.}}, 43:129--159, 2001.

\bibitem[CHRR20]{CHRR20}
Cyrille Chenavier, Clemens Hofstadler, Clemens~G. Raab, and Georg Regensburger.
\newblock Compatible rewriting of noncommutative polynomials for proving
  operator identities.
\newblock In {\em Proceedings of ISSAC 2020}, pages 83--90, 2020.

\bibitem[CIHHP{\etalchar{+}}21]{cvetkovic2021algebraic}
Dragana~S. Cvetkovi{\'c}-Ili{\'c}, Clemens Hofstadler, Jamal Hossein~Poor,
  Jovana Milo{\v{s}}evi{\'c}, Clemens~G. Raab, and Georg Regensburger.
\newblock {Algebraic proof methods for identities of matrices and operators:
  improvements of Hartwig's triple reverse order law}.
\newblock {\em Appl. Math. Comput.}, 409, 2021.

\bibitem[CT05]{candes2005decoding}
Emmanuel~J. Candes and Terence Tao.
\newblock {Decoding by Linear Programming}.
\newblock {\em IEEE Trans.~Inform.~Theory}, 51:4203--4215, 2005.

\bibitem[CW92]{coifman1992entropy}
Ronald~R. Coifman and M.~Victor Wickerhauser.
\newblock {Entropy-Based Algorithms for Best Basis Selection}.
\newblock {\em IEEE Trans. Inform. Theory}, 38:713--718, 1992.

\bibitem[Dan51]{dantzig1951maximization}
George~B. Dantzig.
\newblock Maximization of a linear function of variables subject to linear
  inequalities.
\newblock {\em Activ. Anal. Proc. Alloc.}, 13:339--347, 1951.

\bibitem[DD10]{DD10}
Dragan~S. Djordjevi{\'c} and Neboj{\v{s}}a~{\v{C}}. Din{\v{c}}i{\'c}.
\newblock {Reverse order law for the Moore-Penrose inverse}.
\newblock {\em J. Math. Anal. Appl.}, 361:252--261, 2010.

\bibitem[Den11]{Den11}
Chun~Yuan Deng.
\newblock {A generalization of the Sherman-Morrison-Woodbury formula}.
\newblock {\em Appl. Math. Lett.}, 24:1561--1564, 2011.

\bibitem[dKS13]{d2013heights}
Carlos d’Andrea, Teresa Krick, and Martin Sombra.
\newblock {Heights of varieties in multiprojective spaces and arithmetic
  Nullstellens{\"a}tze}.
\newblock {\em Ann. Sci. {\'E}c. Norm. Sup{\'e}r.}, 46:549--627, 2013.

\bibitem[Don06]{donoho2006compressed}
David~L. Donoho.
\newblock {Compressed Sensing}.
\newblock {\em IEEE Trans. Inform. Theory}, 52:1289--1306, 2006.

\bibitem[EF17]{EF17}
Christian Eder and Jean-Charles Faug{\`e}re.
\newblock {A survey on signature-based algorithms for computing Gr{\"o}bner
  bases}.
\newblock {\em J. Symbolic Comput.}, 80:719--784, 2017.

\bibitem[Fau99]{faugere1999new}
Jean-Charles Faugère.
\newblock {A new efficient algorithm for computing Gr{\"o}bner bases (F4)}.
\newblock {\em J. Pure Appl. Algebra}, 139:61--88, 1999.

\bibitem[FS12]{Faugere-2012-group}
Jean-Charles Faugère and Jules Svartz.
\newblock {Solving polynomial systems globally invariant under an action of the
  symmetric group and application to the equilibria of N vortices in the
  plane}.
\newblock {\em Proceedings of ISSAC 2012}, 2012.

\bibitem[FSEDV16]{Faugere-2016-whomo-2}
Jean-Charles Faugère, Mohab Safey El~Din, and Thibaut Verron.
\newblock {On the complexity of computing Gröbner bases for weighted
  homogeneous systems}.
\newblock {\em J. Symbolic Comput.}, 76:107 -- 141, 2016.

\bibitem[FSS13]{Faugere-2013-gen-minrank}
Jean-Charles Faug{\`e}re, Mohab {Safey El Din}, and Pierre-Jean Spaenlehauer.
\newblock On the complexity of the {G}eneralized {M}in{R}ank {P}roblem.
\newblock {\em J. Symbolic Comput.}, 55:30--58, 2013.

\bibitem[FSS14]{Faugere-2014-sparse}
Jean-Charles Faugère, Pierre-Jean Spaenlehauer, and Jules Svartz.
\newblock {Sparse Gröbner Bases: the Unmixed Case}.
\newblock In {\em {Proceedings of ISSAC 2014}}, pages 178--185, 2014.

\bibitem[GJ79]{Garey-1990-ComputersIntractabilityGuide}
Michael~R. Garey and David~S. Johnson.
\newblock {\em Computers and {{Intractability}}: {{A Guide}} to the {{Theory}}
  of {{NP-Completeness}}}.
\newblock {W. H. Freeman \& Co.}, 1979.

\bibitem[Hof23]{hofstadler-phd}
Clemens Hofstadler.
\newblock {\em {Noncommutative Gröbner bases and automated proofs of operator
  statements}}.
\newblock PhD thesis, Johannes Kepler University Linz, Austria, 2023.
\newblock Available at
  \url{https://resolver.obvsg.at/urn:nbn:at:at-ubl:1-67821}.

\bibitem[Hog13]{hogben2013handbook}
Leslie Hogben.
\newblock {\em {Handbook of Linear Algebra}}.
\newblock CRC press, 2013.

\bibitem[HRR19]{HRR19}
Clemens Hofstadler, Clemens~G. Raab, and Georg Regensburger.
\newblock {Certifying operator identities via noncommutative Gr{\"o}bner
  bases}.
\newblock {\em ACM Commun. Comput. Algebra}, 53:49--52, 2019.

\bibitem[HV22]{Hofstadler-2022-SignatureGroebnerBases}
Clemens Hofstadler and Thibaut Verron.
\newblock {Signature Gr{\"o}bner bases, bases of syzygies and cofactor
  reconstruction in the free algebra}.
\newblock {\em J. Symbolic Comput.}, 113:211--241, 2022.

\bibitem[HV23]{mixed-algebra}
Clemens Hofstadler and Thibaut Verron.
\newblock {Signature Gr\"{o}bner Bases in Free Algebras over Rings}.
\newblock In {\em Proceedings of ISSAC 2023}, pages 298--306, 2023.

\bibitem[HW94]{HW94}
J.~William Helton and John~J. Wavrik.
\newblock Rules for computer simplification of the formulas in operator model
  theory and linear systems.
\newblock In {\em Nonselfadjoint operators and related topics}, pages 325--354.
  Springer, 1994.

\bibitem[{IBM}23]{cplex}
{IBM ILOG CPLEX Optimization Studio}.
\newblock {\em Version}, 22.1(1987--2023), 2023.

\bibitem[Jel05]{Jel05}
Zbigniew Jelonek.
\newblock {On the effective Nullstellensatz}.
\newblock {\em Invent. Math.}, 162:1--17, 2005.

\bibitem[KDC07]{koliha2007moore}
Jerry~J. Koliha, Dragan Djordjevi{\'c}, and Dragana Cvetkovi{\'c}.
\newblock {Moore-Penrose inverse in rings with involution}.
\newblock {\em Linear Algebra Appl.}, 426:371--381, 2007.

\bibitem[Kin19]{kinyon2019proof}
Michael Kinyon.
\newblock Proof simplification and automated theorem proving.
\newblock {\em Philos. Trans. Roy. Soc. A}, 377, 2019.

\bibitem[Kol88]{Kol88}
Janos Kollar.
\newblock {Sharp Effective Nullstellensatz}.
\newblock {\em J. Amer. Math. Soc.}, 1:963--975, 1988.

\bibitem[Lai22]{Lairez-2022-AxiomsForTheory}
Pierre Lairez.
\newblock {Axioms for a theory of signature bases}.
\newblock {\em arXiv preprints}, arXiv:2210.13788v1, 2022.

\bibitem[Li08]{Yua08}
Yuan Li.
\newblock {The Moore–Penrose inverses of products and differences of
  projections in a $C^*$-algebra}.
\newblock {\em Linear Algebra Appl.}, 428:1169--1177, 2008.

\bibitem[LSAZ20]{letterplace}
Viktor Levandovskyy, Hans Sch\"{o}nemann, and Karim Abou~Zeid.
\newblock {\textsc{Letterplace} -- a Subsystem of \textsc{Singular} for
  Computations with Free Algebras via Letterplace Embedding}.
\newblock In {\em {P}roceedings of ISSAC 2020}, pages 305--311, 2020.

\bibitem[Mas21]{mastrolilli2021complexity}
Monaldo Mastrolilli.
\newblock {The Complexity of the Ideal Membership Problem for Constrained
  Problems Over the Boolean Domain}.
\newblock {\em ACM Trans. Algorithms}, 17:1--29, 2021.

\bibitem[May89]{mayr89}
Ernst Mayr.
\newblock {Membership in Polynomial Ideals over $Q$ Is Exponential Space
  Complete}.
\newblock In {\em STACS 89: 6th Annual Symposium on Theoretical Aspects of
  Computer Science Paderborn}, pages 400--406, 1989.

\bibitem[MM82]{mayr-meyer-82}
Ernst~W. Mayr and Albert~R. Meyer.
\newblock {The Complexity of the Word Problems for Commutative Semigroups and
  Polynomial Ideals}.
\newblock {\em Adv. in Math.}, 46:305--329, 1982.

\bibitem[Mor16]{Mor16}
Teo Mora.
\newblock {\em {Solving Polynomial Equation Systems IV: Volume 4, Buchberger
  Theory and Beyond}}, volume 158.
\newblock Cambridge University Press, 2016.

\bibitem[MT17]{mayr2017}
Ernst~W. Mayr and Stefan Toman.
\newblock {\em {Complexity of Membership Problems of Different Types of
  Polynomial Ideals}}, pages 481--493.
\newblock Springer International Publishing, 2017.

\bibitem[MZ93]{mallat1993matching}
St{\'e}phane~G. Mallat and Zhifeng Zhang.
\newblock {Matching Pursuits With Time-Frequency Dictionaries}.
\newblock {\em IEEE Trans. on Signal Process.}, 41:3397--3415, 1993.

\bibitem[PT16]{pitassi2016algebraic}
Tonnian Pitassi and Iddo Tzameret.
\newblock {Algebraic Proof Complexity: Progress, Frontiers and Challenges}.
\newblock {\em arXiv preprints}, arXiv:1607.00443, 2016.

\bibitem[PW00]{potra2000interior}
Florian~A. Potra and Stephen~J. Wright.
\newblock Interior-point methods.
\newblock {\em J. Comput. Appl. Math.}, 124:281--302, 2000.

\bibitem[RP87]{RP87}
Donald~W. Robinson and Roland Puystjens.
\newblock Generalized inverses of morphisms with kernels.
\newblock {\em Linear Algebra Appl.}, 96:65--86, 1987.

\bibitem[RRHP21]{raab2021formal}
Clemens~G. Raab, Georg Regensburger, and Jamal Hossein~Poor.
\newblock Formal proofs of operator identities by a single formal computation.
\newblock {\em J. Pure Appl. Algebra}, 225, 2021.

\bibitem[Sch98]{schrijver1998theory}
Alexander Schrijver.
\newblock {\em Theory of linear and integer programming}.
\newblock John Wiley \& Sons, 1998.

\bibitem[SL20]{SL20}
Leonard Schmitz and Viktor Levandovskyy.
\newblock {Formally Verifying Proofs for Algebraic Identities of Matrices}.
\newblock In {\em International Conference on Intelligent Computer
  Mathematics}, pages 222--236. Springer, 2020.

\bibitem[Ver01]{veroff2001finding}
Robert Veroff.
\newblock {Finding Shortest Proofs: An Application of Linked Inference Rules}.
\newblock {\em J. Automat. Reason.}, 27:123--139, 2001.

\bibitem[Xiu12]{Xiu12}
Xingqiang Xiu.
\newblock {\em {Non-Commutative Gr{\"o}bner Bases and Applications}}.
\newblock PhD thesis, University of Passau, Germany, 2012.
\newblock Available at
  \url{http://www.opus-bayern.de/uni-passau/volltexte/2012/2682/}.

\end{thebibliography}

\end{document}